\documentclass[draftclsnofoot, onecolumn, 11pt]{IEEEtran}

\makeatletter
\def\ps@headings{%
\def\@oddhead{\mbox{}\scriptsize\rightmark \hfil \thepage}%
\def\@evenhead{\scriptsize\thepage \hfil \leftmark\mbox{}}%
\def\@oddfoot{}%
\def\@evenfoot{}}
\makeatother
\usepackage{caption}
\usepackage{subcaption}
\usepackage{mathrsfs}
\usepackage[cmex10]{amsmath}
\usepackage{cite,graphicx,epsfig,wrapfig,amsfonts,amssymb,algorithmic,threeparttable,color,url}
\usepackage{mdwmath,multirow}
\usepackage{mdwtab}

\usepackage{amsthm}
\usepackage[ruled,vlined]{algorithm2e}

\newtheorem{theorem}{Theorem}
\newtheorem{lemma}{Lemma}
\newtheorem{proposition}{Proposition}

\theoremstyle{definition}

\graphicspath{{./}{figures/}}

\ifCLASSINFOpdf

\else

\fi

\begin{document}

\begin{titlepage}
\begin{center}
\vspace*{-2\baselineskip}
\begin{minipage}[l]{7cm}
\flushleft
\includegraphics[width=2 in]{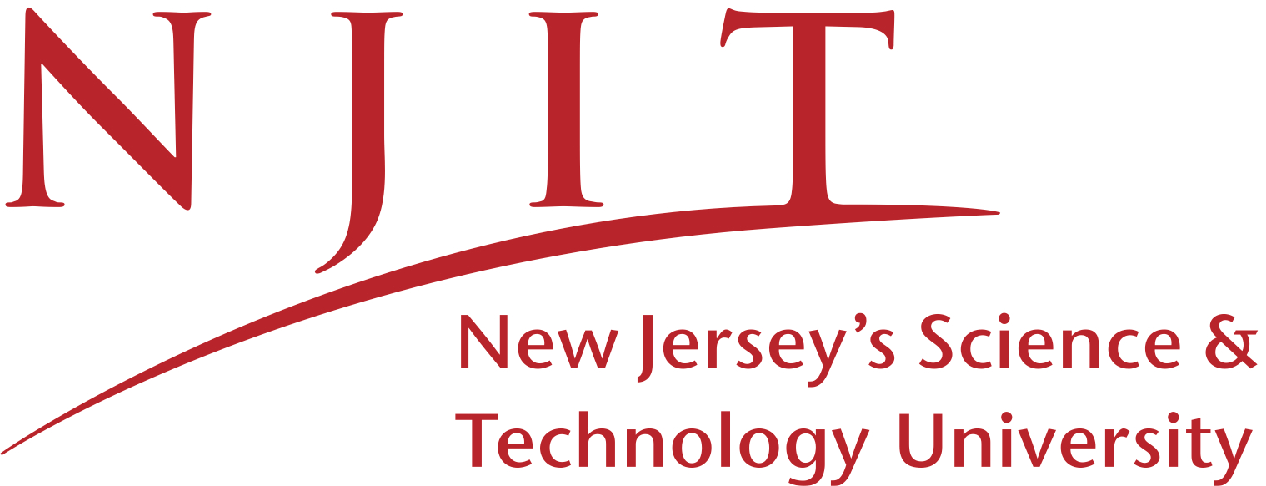}
\end{minipage}
\hfill
\begin{minipage}[r]{7cm}
\flushright
\includegraphics[width=1 in]{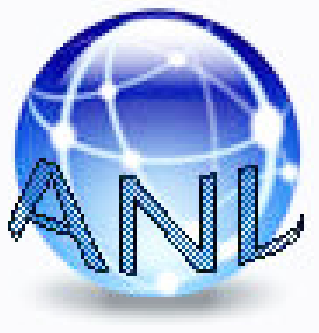}%
\end{minipage}

\vfill

\textsc{\LARGE Provisioning Green Energy for Base Stations in Heterogeneous Networks}\\

\vfill
\textsc{\LARGE Tao Han\\[12pt]
\LARGE NIRWAN ANSARI}\\
\vfill
\textsc{\LARGE TR-ANL-2014-006\\[12pt]
\LARGE Sep. 28, 2014}\\[1.5cm]
\vfill
{ADVANCED NETWORKING LABORATORY\\
 DEPARTMENT OF ELECTRICAL AND COMPUTER ENGINEERING\\
 NEW JERSY INSTITUTE OF TECHNOLOGY}
\end{center}
\end{titlepage}

\title{Provisioning Green Energy for Base Stations in Heterogeneous Networks}
\author{\IEEEauthorblockN{Tao Han, \emph{{Student Member, IEEE}}, and
Nirwan Ansari, \emph{{Fellow, IEEE}}}\\
\IEEEauthorblockA{Department of Electrical and Computer Engineering \\
New Jersey Institute of Technology, Newark, NJ, 07102, USA\\
Email:  \{th36, nirwan.ansari\}@njit.edu}
\thanks{This work was supported in part by NSF under grant no. CNS-1218181 and no. CNS-1320468.}}
\maketitle

\begin{abstract}
Cellular networks are among the major energy hoggers of communication networks, and their contributions to the global energy consumption increase rapidly due to the surges of data traffic. With the development of green energy technologies, base stations (BSs) can be powered by green energy in order to reduce the on-grid energy consumption, and subsequently reduce the carbon footprints. However, equipping a BS with a green energy system incurs additional capital expenditure (CAPEX) that is determined by the size of the green energy generator, the battery capacity, and other installation expenses. In this paper, we introduce and investigate the green energy provisioning (GEP) problem which aims to minimize the CAPEX of deploying green energy systems in BSs while satisfying the QoS requirements of cellular networks. The GEP problem is challenging because it involves the optimization over multiple time slots and across multiple BSs. We decompose the GEP problem into the weighted energy minimization problem and the green energy system sizing problem, and propose a green energy provisioning solution consisting of the provision cost aware traffic load balancing algorithm and the binary energy system sizing algorithm to solve the sub-problems and subsequently solve the GEP problem. We validate the performance and the viability of the proposed green energy provisioning solution through extensive simulations, which also conform to our analytically results.

\end{abstract}
\IEEEpeerreviewmaketitle
\section{Introduction}
\label{sec:introduction}
With the rapid development of radio access techniques and mobile devices, a variety of bandwidth-hungry applications and services such as web browsing, video streaming and social networking are gradually carried through mobile networks, thus leading to an exponential increase of data traffic in cellular networks. The mobile data traffic surges congest cellular networks and degrade the network quality of service (QoS). Heterogeneous networks (HetNets) consisting of both small cell base stations (SBSs) and macro BSs (MBS) are being deployed to offload mobile traffic from MBSs and alleviate the traffic congestion of cellular networks \cite{Yeh:2011:CCE}. SBSs can provide high network capacity for mobile users by capitalizing on their close proximity to mobile users. However, a SBS usually has a limited coverage area. Thus, the number of SBSs will be orders of magnitude larger than that of MBSs for a wide scale network deployment. As a result, the overall energy consumption of cellular networks will keep increasing.

Owing to the direct impact of greenhouse gases on the earth environment and the climate change, there has been a consensus on limiting per-nation $CO_{2}$ emissions \cite{Kyoto:1998}. As a result, governments are likely to regulate the $CO_{2}$ emissions of individual industries in their countries. In this circumstance, mobile service providers may be given a total per-month or per-year energy budgets in terms of $CO_{2}$ emissions \cite{Kwak:2012:GES}. To satisfy the rapidly increasing traffic demands with limited energy budgets, mobile service providers are driven to enhance the energy efficiency of cellular networks.

As green energy technologies advance, green energy such as sustainable biofuels, solar and wind energy can be utilized to power BSs \cite{Han:2014:PMN}. Telecommunication companies such as  Ericsson and Nokia Siemens have designed green energy powered BSs for cellular networks \cite{Ericson:2007:SEU}. By adopting green energy powered SBSs, mobile service providers may save on-grid energy consumption and thus reduce their $CO_{2}$ emissions. For instance, Orange, a french mobile network operator (MNO), has already deployed more than two thousand solar-powered BSs in Africa \cite{Orange:2012:SolarBS}. These BSs serving over 3 million people saved upto 25 million liters of fuel and reduced about 67 million kilogram of $CO_{2}$ emissions in 2011 \cite{Orange:2012:SolarBS}.

Equipping a BSs with a green energy system incurs additional capital expenditures (CAPEX) that are determined by the size of the green energy generator, the battery capacity, and other installation expenses. It is desired to minimize the CAPEX on provisioning green energy while achieving the target QoS. We refer to this problem as the green energy provisioning (GEP) problem. In this paper, we investigate the (GEP) problem. We consider solar energy as the green energy source. Given the per unit cost of the solar panel and the battery capacity, the CAPEX of a BS's green energy system is determined by three variables: the size of the solar panel, the battery capacity, and the cost weight. Here, the cost weight indicates the per unit installation expense of the green energy system on a BS. Given the solar power generation rate and the characteristics of the battery, the size of the solar panel and the battery capacity are determined by the BS's power consumption. Thus, the CAPEX of a BS's green energy system is closely related to the BS's power consumption. A BS's power consumption consists of the static power consumption and the dynamic power consumption \cite{Auer:2011:HME}. The dynamic power consumption is the amount of power consumed for carrying traffic loads. For SBSs, the dynamic power consumption accounts for a very small portion of the total power consumption \cite{Auer:2011:HME}. In other words, the traffic load dependency of SBSs in terms of the power consumption is negligible. Therefore, a SBS's green energy system can be provisioned to ensure the power supplies satisfying the SBS's maximum power demand. Thus, we do not study the green energy provisioning for SBSs.

A MBS's power consumption is, however, highly traffic load dependent. Thus, the power consumption of MBSs can be adjusted by properly balancing traffic loads among BSs. Adapting MBSs' power consumption can change the green energy provision costs and thus reduce the network CAPEX. In order to minimize the network CAPEX, it is desired to reduce the power consumption of the BS which has a large cost weight by optimizing the traffic loads among BSs. Thus, we decompose the GEP problem into two subproblems: the weighted energy minimization (WEM) problem and the green energy system sizing (GESS) problem. A MBS's power consumption depends on its traffic load. Therefore, the WEM problem is solved by designing the provision cost aware traffic load balancing algorithm to optimize the traffic loads among BSs. Given MBSs' traffic loads, the MBSs' energy consumption can be derived. Then, the solar panel size and battery capacity are optimized by solving the GESS problem.

The rest of the paper is organized as follows. In Section \ref{sec:related_works}, we briefly review related works. In Section \ref{sec:sys_model}, we define the system model and formulate the green energy provisioning problem. Section \ref{sec:gep_alg} presents the proposed green energy provisioning solution. Section \ref{sec:simulation} shows the simulation results, and concluding remarks are presented in Section \ref{sec:conclusion}.

\section{Related Works}
\label{sec:related_works}
In this section, we briefly review related works on sizing the green power system and optimizing the green energy utilization in cellular networks.
\subsection{Sizing the green power system}
The process of sizing a green power system involves three basic models: the load model which characterizes energy demands, the battery model which defines the battery capacity and charging characteristics, and the green power generator model which describes the generator capacity \cite{Maghraby:2002:PAP}. Based on these models, three methods can be utilized to determine and evaluate the size of a green power system \cite{Maghraby:2002:PAP}: the loss of load and energy probability method, the fixed autonomy and recharge method, and the Markov chain probabilistic method.
These methods are not applicable to solve the GEP problem because these methods do not optimize the energy demands to minimize the size of green energy system. For the GEP problem, the energy demands of individual MBSs depend on their traffic loads which should be optimized to minimize the network CAPEX. Badawy \emph{et al.} \cite{Badawy:2010:EPS} investigated the energy provisioning problem for solar powered wireless mesh networks, and designed a generic algorithm to incorporate the energy aware routing in the energy provisioning procedure. Although, by incorporating energy aware routing, this method optimizes the energy consumption of wireless nodes, it is designed for wireless mesh networks and cannot be applied to solve the GEP problem. Marsan \emph{et al.} \cite{Marsan:2013:TZG} proposed the concept of zero grid electricity networking in which the cellular networks are powered solely with renewable energy and investigate the problem of dimensioning the power generator capacity and the battery storage. The authors studied the green energy provisioning problem for a single macro BS based on the measurement of the BS power consumption and the renewable energy generation. Our work, however, focuses on optimizing the green energy provision for a collection of BSs in HetNets.

\subsection{Optimizing the green energy utilization}
To optimize the utilization of renewable energy, Ozel \emph{et al.} \cite{Ozel:2011:TEH} proposed to optimize the packet transmission policy for energy harvest wireless nodes. Zhou \emph{et al.} \cite{Zhou:2010:ESA:} proposed the hand over parameter tuning algorithm and the power control algorithm to guide mobile users to access green energy powered BSs. Han and Ansari \cite{Han:2012:ICE} proposed an energy aware cell size adaptation algorithm named ICE, which balances the energy consumption among BSs powered by green energy, and enables more users to be served with green energy. Considering a network with multiple energy supplies, Han and Ansari \cite{Han:2012:OOG} also proposed to optimize the utilization of green energy, and reduce the on-grid energy consumption of cellular networks by the cell size optimization. Assuming the capacity of the green energy system is given, all these solutions are optimizing wireless/cellular networks according to the availability of the green energy. However, for the GEP problem, the capacity of green energy system is to be determined. Therefore, the existing solutions on optimizing the green energy utilization cannot be directly applied to solve the GEP problem.

\section{System Model}
\label{sec:sys_model}
In this paper, we consider a heterogeneous cellular network with multiple MBSs and SBSs. The MBSs are powered by both solar power and grid power while the SBSs are powered by grid power. We focus on optimizing the size of each MBS's green energy system for the downlink transmission. The time horizon is divided into $N$ time slots. In the following analysis, a BS generally refers to a MBS or a SBS.

\subsection{Traffic Model}
We consider the scenario in which MBSs and SBSs are deployed to provide data communications in an area. Denote $\mathcal{B}^{m}$ and $\mathcal{B}^{s}$ as the set of MBSs and SBSs, respectively. We denote $\lambda(x,k)$ and $\nu(x,k)$ as the traffic arrival rate per unit area and the average traffic load at location $x$ in the $k$th time slot, respectively. Here, $\lambda(x,k)$ and $\nu(x,k)$ can be derived based on the statistic traffic data from traffic measurements. For presentation simplicity, we assume there is only one user at location $x$. Assuming that a mobile user at location $x$ is associated with the $j$th BS, then the user's data rate $r_{j}(x)$ can be generally expressed as a logarithmic function of the perceived signal to interference plus noise ratio, $SINR_{j}(x)$, according to the Shannon –Hartley theorem \cite{Kim:2012:DOU},
\begin{equation}
\label{eq:user_rate}
r_{j}(x)=log_{2}(1+SINR_{j}(x)).
\end{equation}
Here,
\begin{equation}
\label{eq:user_SINR}
SINR_{j}(x)=\frac{P_{j}g_{j}(x)}{\sigma^{2}+\sum_{h \in \mathcal{B},h\neq j}P_{h}g_{h}(x)},
\end{equation}
where $P_{j}$ is the transmission power of the $j$th BS, $\sigma^{2}$ denotes the noise power level, and $g_{j}(x)$ is the channel gain between the $j$th BS and the user at location $x$. Here, the channel gain reflects only the slow fading including the path loss and the shadowing. For the energy provisioning purpose, the channel gain is measured at a large time scale, and thus fast fading is not considered. The average traffic load density at location $x$ in the $j$th BS is
\begin{equation}
\label{eq:bs_point_load}
\varrho_{j}(x,k)=\frac{\lambda(x,k)\nu(x,k)\eta_{j}(x)}{r_{j}(x)}.
\end{equation}
Here, $\eta_{j}(x)$ is an indicator function. If $\eta_{j}(x)=1$, the user at location $x$ is served by the $j$th BS; otherwise, the user is not served by the BS.
Assuming mobile users are uniformly distributed in the area and denoting $\mathcal{A}$ as the coverage area of all the BSs, the traffic load on the $j$th BS can be expressed as
\begin{equation}
\label{eq:bs_load}
\rho_{j}(k)=\int_{x \in \mathcal{A}}\varrho_{j}(x,k)dx.
\end{equation}
This value of $\rho_{j}(k)$ indicates the fraction of time BS $j$ is busy in the $k$th time slot.

\subsection{Energy Model}
In the network, MBSs are powered by both green energy and on-grid energy. Since we aim to investigate the green energy provision for MBSs, we assume SBSs are powered by on-grid power. The MBS's power consumption consists of two parts: the static power consumption and the dynamic power consumption \cite{Arnold:2010:PSM}. The static power consumption is the power consumption of a MBS without any traffic load. The dynamic power consumption refers to the additional power consumption caused by carrying traffic loads in the MBS, which can be well approximated by a linear function of the traffic load \cite{Arnold:2010:PSM}. Denote $p^{s}_{j}$ as the static power consumption of the $j$th MBS. Then, the $j$th MBS's power consumption in the $k$th time slot can be expressed as
\begin{equation}
\label{eq:bs_power_consumption}
p_{j}(k)=\beta_{j}\rho_{j}(k)+p^{s}_{j}.
\end{equation}
Here, $\beta_{j}$ is a linear coefficient which reflects the relationship between the traffic load and the dynamic power consumption in the $j$th MBS.

Denote $e_{j}(k)$ as the green energy capacity per unit area of solar panel in the $j$th MBS in the $k$th time slot. We define $S_{j}$, $B^{max}_{j}$, and $B^{min}_{j}$ as the solar panel size, the battery capacity, and the minimum permitted battery energy of the $j$th MBS's green power system, respectively. We adopt the linear charge model for the solar power system \cite{Badawy:2010:EPS}. Then, the $j$th MBS's battery energy in the $k$th time slot can be expressed as
\begin{equation}
\label{eq:battery_energy}
b_{j}(k)=\min\{\max\{b_{j}(k-1)+e_{j}(k)S_{j}-p_{j}(k),B^{min}_{j}\},B^{max}_{j}\}.
\end{equation}
$b_{j}(k)$ depends on the battery energy in the $(k-1)$th time slot and the energy generation and consumption in the current time slot. For the consideration of the safety and the battery life, the battery is not allowed to be discharged below $B^{min}_{j}$. In other words, if $b_{j}(k)\leq B^{min}_{j}$, the charge controller disconnects the $j$th MBS from the battery and pulls power from power gird. For simplicity, we assume $B^{min}_{j}=0$ in this paper. The battery also cannot be charged beyond its capacity, $B^{max}_{j}$.

The cost of the green energy system is determined by the solar panel size and the battery capacity. In this paper, we adopt a simple linear model to reflect the cost of the green energy system versus the solar panel size and the battery capacity as follows:
\begin{equation}
\label{eq:energy_cost}
f_{j}(S_{j},B^{max}_{j})=\phi_{s}S_{j}+\phi_{b} B^{max}_{j}.
\end{equation}
Here, $\phi_{s}$ and $\phi_{b}$ indicate the cost per unit area of solar panel and per unit battery capacity, respectively. In addition, equipping a MBS with green energy also incurs installation expenses including the labor costs and the space rental costs which are the expense on leasing the space for installing the green energy system. The per unit energy system installation cost may be different for the MBSs at different locations. In this paper, we assume the locations of MBSs are pre-determined.
We define $w_{j}$ as the cost weight of installing per unit green energy system in the $j$th MBS.
The CAPEX of the $j$th MBS's green energy system is $w_{j}f_{j}(S_{j},B^{max}_{j})$.

\subsection{Problem Formulation}
\label{sec:problem_formulation}
The CAPEX of MBSs' green energy systems depend on the power consumption of the MBSs. In order to minimize the CAPEX, it is desirable to offload as much traffic load from MBSs to SBSs as possible. The aggressive traffic offloading may lead to traffic congestion in SBSs, and thus downgrading the QoS of the network. Therefore, when optimizing the green energy provision, the traffic offloading should be properly considered to ensure the QoS of the network.

We assume that traffic arrival processes at individual users are independent and follow Poisson distributions. Then, the traffic arrival in the $j$th BS, which is the sum of the traffic arrivals toward all users in its coverage area, is also a Poisson process. The required service time for a user at location $x$ in the $j$th BS is
\begin{equation}
\label{eq:required_time}
\theta_{j}= \frac{\nu(x,k)}{r_{j}(x)}.
\end{equation}
Since $\nu(x,k)$ follows a general distribution, the user's required service time is also a general distribution. Hence, a BS's service rate follows a general distribution. Therefore, a BS's downlink transmission process realizes a M/G/1 processor sharing queue, in which multiple users share the BS's downlink radio resource \cite{Kleinrock:1976:QS}.

In mobile networks, various downlink scheduling algorithms have been proposed to enable proper sharing of the limited radio resource in a BS \cite{Capozzi:2013:DPS}. These algorithms are designed to maximize the network capacity, enhance the fairness among users, or provision QoS services. According to the scheduling algorithm, users are assigned different priorities on sharing the downlink radio resource. We assume that during the traffic balancing process, users' data rates do not change. As a result, users in different priority groups perceive different average waiting time. Since traffic arrives at a BS according to Possion arrival statistics, the allowed variation in the average waiting times among different priority groups is constrained by the Conservation Law \cite{Kleinrock:1976:QS}. The integral constraint on the average waiting time in the $j$th BS in the $k$th time slot can be expressed as
\begin{equation}
\label{eq:conservation_law}
\bar{L}_{j}=\frac{\rho_{j}(k)E(\theta_{j}^{2})}{2(1-\rho_{j}(k))}.
\end{equation}
This indicates that given the users' required service time in the $j$th BS, if the scheduling algorithm gives some users higher priority and reduces their average waiting time, it will increase the average waiting time of the other users. Therefore, $\bar{L}_{j}$ generally reflects the $j$th BS's performance in terms of users' average waiting time.
Since $E(\theta_{j}^{2})$ mainly reflects the traffic characteristics, we assume that $E(\theta_{j}^{2})$ is roughly constant during a user association process and define $\vartheta_{j}=\frac{E(\gamma_{j}^{2})}{2}$. Thus, we adopt
\begin{equation}
\label{eq:latency_indicator}
\mu(\rho_{j}(k))=\frac{\vartheta_{j}\rho_{j}(k)}{1-\rho_{j}(k)}
\end{equation}
as a general latency indicator for the $j$th BS. A smaller $\mu(\rho_{j}(k))$ indicates that the $j$th BS introduces less latency to its associated users. For simplicity, we use $\mu_{j}(k)$ to represent $\mu(\rho_{j}(k))$. We utilize $\mu_{j}(k)$ as the QoS indicator for the $j$th BS in the $k$th time slot. To ensure the QoS of the network, $\mu_{j}(k)$ should be less than a threshold, $\zeta$.

Then, the green energy provisioning (GEP) problem can be formulated as
\begin{eqnarray}
\label{eq:obj_energy_provision}
\min_{(S_{j},B_{j}^{max},\forall j\in\mathcal{B}^{m})}&&\sum_{j\in\mathcal{B}^{m}}w_{j}f_{j}(S_{j},B_{j}^{max}) \\
subject\;to:&&\mu_{j}(k)\leq \zeta, \;\forall j\in \mathcal{B}^{m}\cup \mathcal{B}^{s};\nonumber\\
&& b_{j}(k-1)+e_{j}(k)S_{j}\geq \alpha(k)p_{j}(k),\nonumber\\
&&\forall j \in\mathcal{B}^{m}; \nonumber\\
&& 0\leq\rho_{j}(k)\leq 1-\epsilon, \;\forall j\in \mathcal{B}^{m}\cup \mathcal{B}^{s},\nonumber\\
&& \forall k \in\{1, 2, \cdots, N\}.
\end{eqnarray}
There are three constraints for the GEP problem. The first constraint imposes the latency ratios of all BSs to be no larger than $\zeta$. Here, $\zeta$ should be properly selected to ensure the feasibility of the GEP problem.
The second constraint imposes individual MBSs' green power supplies not to be less than its green power demand. Here, green power is defined as power generated from green energy. $0\leq\alpha(k)\leq 1$ is a system parameter that defines the percentage of the power consumption that should be pulled from the MBS's green energy system. $\alpha(k)$ can be selected by the mobile network operators when provisioning the green energy system. A larger $\alpha(k)$ usually results in a higher CAPEX.
The third constraint is to ensure the queuing system to be stable by restricting the traffic load in individual BSs to be less than 1. Here, $\epsilon$ is an arbitrary small positive.

Given the BS deployment and the traffic load statistic, the lower bound of $\zeta$ can be derived in two steps. First, by solving the QoS bound (QB) problem expressed as
\begin{eqnarray}
\label{eq:zeta_boud}
\min_{\rho_{j}(k)} && \max_{j\in\mathcal{B}^{m}\cup\mathcal{B}^{s}} \mu_{j}(k)\\
\label{eq:zeta_constraint}
subject\; to: && 0\leq\rho_{j}(k)\leq 1-\epsilon,
\end{eqnarray}
we can derive the lower bound of $\zeta$ in the $k$th time slot, which is denoted as
\begin{equation}
\label{eq:low_bound}
\mu^{*}(k)=\frac{\vartheta_{j}\rho_{j}^{*}(k)}{1-\rho_{j}^{*}(k)}, \; j=\arg\max_{l\in\mathcal{B}^{m}\cup\mathcal{B}^{s}} \mu_{l}(k).
\end{equation}
Here, $\rho_{j}^{*}(k)$ is the $j$th BS's optimal traffic load derived by solving the QB problem in the $k$th time slot. Then, $\zeta$'s lower bound $\zeta^{*}=\max_{k\in\{1,2,\cdots,N\}}\mu^{*}(k)$. To ensure the GEP problem to be feasible, $\zeta\geq\zeta^{*}$. Similar as the parameter $\alpha(k)$, $\zeta$ is predetermined by the mobile network operator for the green energy provision.

\section{The Green Energy Provisioning Solution}
\label{sec:gep_alg}
Solving the GEP problem is equivalent to determining the optimal solar panel sizes and battery capacities for MBSs. Since the green power systems are provisioned to operate the MBSs during a certain time period (multiple time slots), the solar panel sizes and the battery capacities are determined to satisfy MBSs' green power demands over the time slots. Since battery energy in a time slot depends on that in the previous slots, the optimal solar panel size and battery capacity for a MBS is determined by the MBS's green power demands in multiple time slots. Within a time slot, e.g., the $k$th time slot, the green power demands in a MBS, e.g., the $j$th MBS, depend on its traffic load $\rho_{j}(k)$ and the parameter $\alpha(k)$. Thus, solving the GEP problem involves optimizing their traffic load in multiple time slots. Owing to the complex coupling of network optimization in multiple time slots, it is very challenging to solve the GEP problem.

\subsection{Problem decomposition}
In order to solve the GEP problem, we decompose the GEP problem into two sub-problems: the weighted energy minimization (WEM) problem and the green energy system sizing (GESS) problem. In this way, we decouple the interdependence of network optimization in multiple time slots. The WEM problem optimizes the network's weighted energy cost in individual time slots while the GESS problem optimizes solar panel sizes and battery capacities for individual SBSs according to their energy demands over multiple time slots.

On optimizing the solar panel size, we assume that the $j$th MBS's initial battery energy, $b_{j}(0)$, is zero, and the green energy consumed by the $j$th MBS is all generated from its solar panel. Thus,
\begin{equation}
\label{eq:decomp_sloar_1}
\sum_{k\in\{1,2,\cdots,N\}}e_{j}(k)S_{j}\geq \sum_{k\in\{1,2,\cdots,N\}}\alpha(k)p_{j}(k).
\end{equation}
Considering all MBSs and their weights,
\begin{align}
\label{eq:decomp_sloar_2}
\sum_{k\in\{1,2,\cdots,N\}}&\sum_{j\in\mathcal{B}^{m}}w_{j}e_{j}(k)S_{j}\geq \nonumber\\
&\sum_{k\in\{1,2,\cdots,N\}}\sum_{j\in\mathcal{B}^{m}}w_{j}\alpha(k)p_{j}(k).
\end{align}
E.q. (\ref{eq:decomp_sloar_2}) can be rewritten as
\begin{align}
\label{eq:decomp_sloar_3}
\sum_{j\in\mathcal{B}^{m}}(w_{j}S_{j}&\sum_{k\in\{1,2,\cdots,N\}}e_{j}(k))\geq \nonumber\\
 &\sum_{k\in\{1,2,\cdots,N\}}\sum_{j\in\mathcal{B}^{m}}w_{j}\alpha(k)p_{j}(k).
\end{align}
$e_{j}(k),\;\forall j \in\mathcal{B}^{m},\;\forall k\in\{1,2,\cdots,N\}$ is derived based on the statistical solar power data and is considered as a constant. We assume all MBSs have the similar geolocations. Thus, $\sum_{k\in\{1,2,\cdots,N\}}e_{j}(k)=\sum_{k\in\{1,2,\cdots,N\}}e_{i}(k), \forall i, j \in \mathcal{B}^{m}$. Therefore, minimizing $\sum_{k\in\{1,2,\cdots,N\}}\sum_{j\in\mathcal{B}^{m}}w_{j}\alpha(k)p_{j}(k)$ is necessary to minimize $\sum_{j\in\mathcal{B}^{m}}w_{j}S_{j}$. Since the traffic arrival is a Poisson process, a MBS's traffic load in different time slots are independent. Thus, the MBS's energy consumption in different time slots is independent. Therefore, minimizing $\sum_{k\in\{1,2,\cdots,N\}}\sum_{j\in\mathcal{B}^{m}}w_{j}\alpha(k)p_{j}(k)$ is equivalent to minimizing $\sum_{j\in\mathcal{B}^{m}}w_{j}\alpha(k)p_{j}(k),\;\forall k \in \{1,2,\cdots,N\}$.

Since the MBSs share the similar geolocation, $e_{j}(k)=e_{i}(k),\;\forall i,j\in\mathcal{B},\;\forall k\in\{1,2,\cdots,N\}$. Therefore, the time slots in which solar power is zero are the same for all the MBSs. Denote $\mathcal{K}^{b}$ as the set of these time slots. During these time slots, battery energy is utilized to satisfy the MBSs' demands for green energy. Thus,
\begin{equation}
\label{eq:decom_battery_1}
B^{max}_{j}\geq \sum_{k\in\mathcal{K}^{b}}\alpha(k)p_{j}(k).
\end{equation}
Considering all the MBSs and their cost weights,
\begin{equation}
\label{eq:decom_battery_2}
\sum_{j\in\mathcal{B}^{m}}w_{j}B^{max}_{j}\geq \sum_{k\in\mathcal{K}^{b}}\sum_{j\in\mathcal{B}^{m}}w_{j}\alpha(k)p_{j}(k).
\end{equation}
Since MBSs' energy consumption in different time slots is independent, minimizing $\sum_{j\in\mathcal{B}^{m}}w_{j}\alpha(k)p_{j}(k), \; \forall k\in\mathcal{K}^{b}$ is necessary in order to minimize $\sum_{j\in\mathcal{B}^{m}}w_{j}B^{max}_{j}$.

Based on the above analysis, the WEM problem can be expressed as
\begin{eqnarray}
\label{eq:obj_wem}
\min_{(\rho_{j}(k),\forall j\in\mathcal{B}^{m})}&&\sum_{j\in\mathcal{B}^{m}}w_{j}\alpha(k)p_{j}(k) \\
subject\;to:&&\mu_{j}(k)\leq \zeta, \;\forall j\in \mathcal{B}^{m}\cup \mathcal{B}^{s},\nonumber\\
&& 0\leq\rho_{j}(k)\leq 1-\epsilon, \;\forall j\in \mathcal{B}^{m}\cup \mathcal{B}^{s}.\nonumber\\
\end{eqnarray}

Given the $j$th MBS's energy consumption in all time slots, the GESS problem can be expressed as
\begin{eqnarray}
\label{eq:obj_gess}
\min_{(S_{j},B_{j}^{max})}&&f_{j}(S_{j},B_{j}^{max}) \\
subject\;to:&& b_{j}(k-1)+e_{j}(k)S_{j}\geq \alpha(k)p_{j}(k),\nonumber\\
&& \forall j \in\mathcal{B}^{m}, \forall k \in\{1, 2, \cdots, N\}.
\end{eqnarray}

\subsection{The provision cost aware traffic load balancing}
Since the WEM problem minimizes MBSs' weighted power consumption within a time slot, we use $\mu_{j}$, $\rho_{j}$, and $\alpha$ instead of $\mu_{j}(k)$, $\rho_{j}(k)$, and $\alpha(k)$. Since $\vartheta_{j}$ is a constant within a time slot, we assume $\vartheta_{j}=1$ for presentation simplicity.
For the WEM problem, since $0<\zeta<\infty$, when $\mu_{j}\leq \zeta$, $\rho_{j}\leq \frac{\zeta}{1+\zeta}\leq 1-\epsilon$. Since $\rho_{j}>0$, $\mu_{j}(k)\leq \zeta$ indicates $0\leq\rho_{j}\leq 1-\epsilon$. Therefore, the second inequality constraint of the WEM problem can be eliminated. We then apply Lagrangian dual decomposition to design a provision cost aware traffic load balancing algorithm solving the WEM problem.

Let $\mathcal{B}$=$\mathcal{B}^{m}\cup\mathcal{B}^{s}$ and $w_{j}=0$, $\forall j \in
\mathcal{B}^{s}$.
We introduce a Lagrangian multiplier vector, $\boldsymbol{\upsilon}=(\upsilon_{1},\cdots,\upsilon_{j},\cdots,\upsilon_{|\mathcal{B}|})$. The Lagrangian function of the WEM problem is

\begin{align}
\label{eq:obj_atlo}
L(\boldsymbol{\rho},\boldsymbol{\upsilon})&=\sum_{j\in\mathcal{B}}\alpha w_{j}p_{j} -\upsilon_{j}(\frac{\zeta}{1+\zeta}-\rho_{j})\nonumber\\
&=\sum_{j\in\mathcal{B}}(\alpha w_{j}\beta_{j}+\upsilon_{j})\rho_{j} +\alpha w_{j}\beta_{j}p^{s}_{j}-\upsilon_{j}\frac{\zeta}{1+\zeta}.
\end{align}
Here, $\boldsymbol{\rho}=(\rho_{1},\cdots,\rho_{j},\cdots,\rho_{|\mathcal{B}|})$. Since
\begin{equation}
\label{eq:bs_load_discrete}
\rho_{j}=\sum_{x\in\mathcal{A}}\frac{\lambda_{i}\nu_{i}\eta_{j}(x)}{r_{j}(x)},
\end{equation}
the dual function is given as
\begin{equation}
\label{eq:dual_function}
g(\boldsymbol{\upsilon})= \inf_{\boldsymbol{\eta}}h(\boldsymbol{\upsilon},\boldsymbol{\eta})+\sum_{j\in\mathcal{B}}\alpha w_{j}\beta_{j}p^{s}_{j}-\upsilon_{j}\frac{\zeta}{1+\zeta}.
\end{equation}
where
\begin{equation}
\label{eq:user_side_problem}
h(\boldsymbol{\upsilon},\boldsymbol{\eta})=
\sum_{x\in\mathcal{A}}\sum_{j\in\mathcal{B}}(\alpha w_{j}\beta_{j}+\upsilon_{j})\frac{\lambda_{i}\nu_{i}\eta_{j}(x)}{r_{j}(x)}
\end{equation}
and

\begin{equation}
\boldsymbol{\eta}=\{\eta_{j}(x)| j\in \mathcal{B}, x \in \mathcal{A}\}
\end{equation}
The dual problem is
\begin{eqnarray}
\label{eq:dual_problem}
\max_{\boldsymbol{\upsilon}}&&g(\boldsymbol{\upsilon}) \\
subject\;to:&& \upsilon_{j}\geq 0, \forall j\in\mathcal{B}.
\end{eqnarray}

The provision cost aware traffic load balancing algorithm solves the dual problem and thus addresses the WEM problem. The proposed algorithm includes two parts: the traffic redirect algorithm and the traffic load update algorithm.

The traffic redirect algorithm derives $\eta_{j}(x)$ that minimizes $h(\boldsymbol{\upsilon},\boldsymbol{\eta})$ while the traffic load update algorithm finds the optimal $\boldsymbol{\upsilon}$ that maximizes $g(\boldsymbol{\upsilon})$. The provision cost aware traffic load balancing involves multiple iterations. We denote $\boldsymbol{\upsilon}^{t}=\{\upsilon^{t}_{j}|i\in\mathcal{B}\}$ as the Lagrangian multiplier in the $t$th iteration.
\subsubsection{The traffic redirect algorithm} this algorithm calculates the downlink data rates from all BSs based on the SINR measurements for a user at a location. The traffic to the a user at location $x$ is redirected to the $j^{*}$th BS according to the following traffic redirect rule:
\begin{equation}
\label{eq:user_selection}
j^{*}=\arg\min_{j\in\mathcal{B}}(\alpha w_{j}\beta_{j}+\upsilon^{t}_{j})\frac{\lambda_{i}\nu_{i}\eta_{j}(x)}{r_{j}(x)}.
\end{equation}
\begin{lemma}
\label{thm:user_side_min}
Given $\boldsymbol{\upsilon}^{t}$, the traffic redirect algorithm minimizes $h(\boldsymbol{\upsilon}^{t},\boldsymbol{\eta})$.
\end{lemma}
\begin{proof}
\label{prf:user_side_min}
Since a user can only associate with one BS, if $\eta_{j^{*}}(x)=1$, $\forall j\in\mathcal{B}$ and $j\neq j^{*}$, $\eta_{j}(x)=0$. Denote $\boldsymbol{\eta^{*}}$ as the traffic redirection derived by traffic redirect algorithm. Assume $\boldsymbol{\eta}$ is an arbitrary traffic redirection that $\boldsymbol{\eta}\neq\boldsymbol{\eta^{*}}$.
\begin{align}
\label{eq:user_side_min}
&h(\boldsymbol{\upsilon}^{t},\boldsymbol{\eta^{*}})-h(\boldsymbol{\upsilon}^{t},\boldsymbol{\eta})\\
&=\sum_{x\in\mathcal{A}}[(\alpha w_{j^{*}}\beta_{j^{*}}+\upsilon^{t}_{j^{*}})\frac{\lambda_{i}\nu_{i}\eta_{j^{*}}(x)}{r_{j^{*}}(x)}-(\alpha w_{j}\beta_{j}+\upsilon^{t}_{j})\frac{\lambda_{i}\nu_{i}\eta_{j}(x)}{r_{j}(x)}]
\end{align}
because
\begin{equation}
\label{eq:user_selection}
j^{*}=\arg\min_{j\in\mathcal{B}}(\alpha w_{j}\beta_{j}+\upsilon^{t}_{j})\frac{\lambda_{i}\nu_{i}\eta_{j}(x)}{r_{j}(x)}.
\end{equation}
$h(\boldsymbol{\upsilon},\boldsymbol{\eta^{*}})-h(\boldsymbol{\upsilon}^{t},\boldsymbol{\eta})\leq 0$. Thus, the lemma is proved.
\end{proof}

\subsubsection{traffic load update algorithm} given the traffic redirection, $\boldsymbol{\eta^{*}}$, the traffic load update algorithm measures BSs' traffic load and update the Lagrangian multiplier to maximize $g(\boldsymbol{\upsilon})$. Denote $\rho^{t}_{j}$ as the $j$th BS's traffic load after the $t$th iteration. The multiplier in the $j$th BS in the $(t+1)$th iteration is updated as
\begin{equation}
\label{eq:lag_multiplier}
\upsilon^{t+1}=\upsilon^{t}+\delta^{t}(\rho^{t}_{j}-\frac{\zeta}{1+\zeta}).
\end{equation}
Here, $\delta^{t}>0$ is a dynamically selected step size that ensures the convergence of the iterations between users and BSs. $\delta^{t}$ is chosen base on
\begin{equation}
\label{eq:stepsize_rule}
\delta^{t}=\gamma^{k}\frac{g(\boldsymbol{\upsilon}^{t})-g(\hat{\boldsymbol{\upsilon}})+\varepsilon^{t}}{\|\rho_{j}^{t}-\frac{\zeta}{1+\zeta}\|^{2}}
\end{equation}
where $0<\underline{\gamma}\leq\gamma^{k}\leq\overline{\gamma}<2$, $\underline{\gamma}$ and $\overline{\gamma}$ are some scalar \cite{Bertsekas:20049:CVX} and $\varepsilon^{t}$ is updates according to
\begin{equation}
\label{eq:update_approximate}
\varepsilon^{t}=\left\{
\begin{aligned}
&a\varepsilon^{t},&g(\boldsymbol{\upsilon}^{t+1})\leq g(\boldsymbol{\upsilon}^{t})\\
&\max(b\varepsilon^{t},\varepsilon), &g(\boldsymbol{\upsilon}^{t+1})> g(\boldsymbol{\upsilon}^{t}),
\end{aligned}
\right.
\end{equation}
where $a$, $b$ and $\varepsilon$ are fixed positive constants with $a\geq1$ and $b<1$.
In Eq. (\ref{eq:stepsize_rule}),
$\boldsymbol{\hat{\upsilon}}=\{\hat{\upsilon}_{j}|j\in\mathcal{B}\}$ is an estimation of the optimal Lagrangian multiplier as
\begin{equation}
\label{eq:obj_estimation}
\hat{\boldsymbol{\upsilon}}= \arg\min_{(\boldsymbol{\upsilon}^{m}, 0\leq m\leq t)} g(\boldsymbol{\upsilon}^{m}).
\end{equation}
\begin{proposition}
\label{thm:bound_proposition}
There exists some scalar $c$ such that
\begin{equation}
\label{eq:propostion}
\sup\{\|q(\boldsymbol{\upsilon})\|\;|\;q(\boldsymbol{\upsilon})\in \partial{g(\boldsymbol{\upsilon}^{t})}, \forall t\geq 0\}\leq c.
\end{equation}
\end{proposition}
\begin{proof}
\begin{equation}
\partial{g(\boldsymbol{\upsilon}^{t})}=\inf_{\mathcal{\eta}}\sum_{x\in\mathcal{A}}\frac{\lambda_{i}\nu_{i}\eta_{j}(x)}{r_{j}(x)} - \frac{\zeta}{1+\zeta}.
\end{equation}
Because $\eta_{j}(x)=\{0,1\}$, $\sum_{x\in\mathcal{A}}\frac{\lambda_{i}\nu_{i}\eta_{j}(x)}{r_{j}(x)}$ is bounded. Thus, the subgradient of the dual problem is bounded:
\begin{equation}
\label{eq:propostion}
\sup\{\|q(\boldsymbol{\upsilon})\|\;|\;q(\boldsymbol{\upsilon})\in \partial{g(\boldsymbol{\upsilon}^{t})}, \forall t\geq 0\}\leq c.
\end{equation}
\end{proof}
\begin{theorem}
\label{thm:ua_converge}
Assume that $\delta^{t}$ is determined by the dynamic step size rule in Eq. (\ref{eq:stepsize_rule}) with the adjustment procedures in Eqs. (\ref{eq:update_approximate}) and (\ref{eq:obj_estimation}). If $g(\boldsymbol{\upsilon}^{*})<\infty$,
\begin{equation}
\label{eq:thm_converge}
\sup_{t\geq 0} g(\boldsymbol{\upsilon}^{t})\geq g(\boldsymbol{\upsilon}^{*}) - \varepsilon,
\end{equation}
where $\boldsymbol{\upsilon}^{*}$ denotes the optimal Lagrangian multiplier.
\end{theorem}
\begin{proof}
\label{prf:ua_converge}
Based on Proposition \ref{thm:bound_proposition}, the dual problem satisfies the necessary condition of Proposition 6.3.6 in \cite{Bertsekas:20049:CVX}. The theorem is proved by applying this proposition.
\end{proof}

After $\boldsymbol{\upsilon}$ converges, the optimal traffic load balancing is derived according to the traffic redirect algorithm, based on which we obtain the optimal traffic load $\boldsymbol{\rho}^{*}$ and thus calculate the BSs' energy consumption in the time slot.
\subsection{The green energy system sizing}
After solving the WEM problem for all the time slots, we obtain individual MBSs' energy consumption in each time slot. Based on the energy consumption, we solve the GESS problem to derive the optimal solar panel size and battery capacity for MBSs.
\begin{figure}
\centering
\includegraphics[scale=0.8]{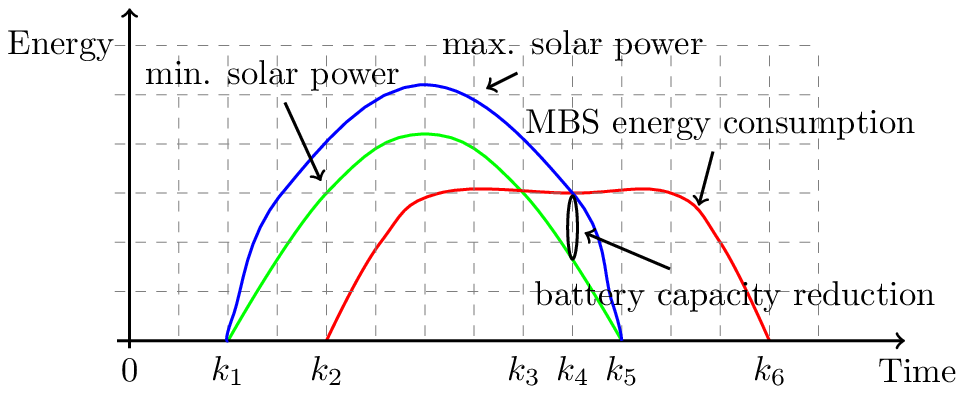}
\caption{An illustration of the solar panel size and the battery capacity.}
\label{fig:gess_eg}
\end{figure}
Fig. \ref{fig:gess_eg} shows an example of the solar power generation and the green power demand in a MBS. The solar power generation starts in the $k_{1}$th time slot and ends in the $k_{5}$th time slot. The MBS is activated in the $k_{2}$th time slot and turned off in the $k_{6}$th time slot. In order to power the MBS, the solar power generation should at least equal to the MBS's green power consumption. A MBS's green power consumption equals to the MBS's total power consumption multiplied by the percentage of power pulled from green energy.
We define the minimum solar panel size as the solar panel size with which the solar power generation equals to the MBS's green power consumption.
The green plot in the figure indicates the solar power generation with the minimum panel size. In this case, the MBS's green power comes from both the solar panel and the battery in the $k_{4}$th time slot while the battery is responsible for the power supplies in the $k_{5}$th and the $k_{6}$th time slot.
The battery capacity should at least equal to the MBS's green power consumption from the $k_{4}$th to the $k_{6}$th time slot minus the solar power generation in the $k_{4}$th time slot.
\begin{lemma}
\label{thm:panel_battery}
On powering the SBS, increasing the solar panel size does not increase the required battery capacity.
\end{lemma}
\begin{proof}
\label{prf:panel_battery}
The battery is responsible for the power supplies during the time slots in which solar power is less than the MBS's power consumption. Given the solar energy generation rate, increasing the solar panel size does not increase the MBS's energy consumption, and thus does not increase the required battery capacity.
\end{proof}

In some cases, increasing the solar panel size enables the reduction of the required battery capacity.
For example, as shown in the blue plot in Fig. \ref{fig:gess_eg}, the increase of the solar panel size increases the solar power generation. As a result, the MBS's energy consumption in the $k_{4}$th time slot is fully covered by solar power. The battery is only responsible for the power supplies from the $k_{5}$th to $k_{6}$th time slot. Thus, the battery capacity can be reduced.
However, when the solar panel size is large enough, a further increase of the solar panel size does not decrease the required battery capacity. For example, assume the MBS's solar power generation is shown in the blue plot in Fig. \ref{fig:gess_eg}. Further increase of the solar panel size does not reduce the required battery capacity because the solar power generation rate is zero in the $k_{5}$th time slot. We define the maximum solar panel size as the solar panel size with which a further increase of the panel size does not decrease the required battery capacity. Denote $S^{max}_{j}$ as the $j$th MBS's maximum solar panel size.
\begin{equation}
\label{eq:max_size}
S^{max}_{j}=\left\lceil\max_{k \in \{l|e_{j}(l)>0,\;l\in\{1,2,\cdots,N\}\}}\frac{\alpha(k)p_{j}(k)}{e_{j}(k)}\right \rceil.
\end{equation}
Here,  $\lceil x \rceil$ denotes the smallest integer that is greater than or equal to $x$. Denote $S^{min}_{j}$ as the $j$th SBS's minimum solar panel size. Let $m_{j}=\arg\max_{k \in \{l|p_{j}(l)>0,\;l\in\{1,2,\cdots,N\}\}} k$
\begin{equation}
\label{eq:min_size}
S^{min}_{j}=\left\lceil\frac{\sum_{k\in\{1,2,\cdots,m_{j}\}}p_{j}(k)}{\sum_{k\in\{1,2,\cdots,m_{j}\}}e_{j}(k)}\right\rceil.
\end{equation}
Solving the GESS problem involves the trade off between the solar panel size and the battery capacity. We apply the binary search method to find the optimal solar panel size, and then derive the corresponding battery capacity.
Given the $j$th MBS's solar panel size and the solar energy generation rates, the $j$th MBS's solar power generation in an individual time slot is calculated. Given the MBS's green power consumption, the battery capacity is derived to guarantee that sufficient energy is stored to satisfy the MBS's energy demand in each time slot. If the current solar panel size cannot sufficiently charge the battery to power the MBSs during the time slots in which the solar power generation is less than the power consumption, we set the battery capacity to be infinity. As a result, the cost of the green energy system will be infinity. Thus, the binary energy system sizing (BESS) algorithm increases the solar panel size. Denote $S^{tmp}_{j}$ and $B^{tmp}_{j}$ as the $j$th MBS's intermediate solar panel size and battery capacity, respectively. The pseudo code of the (BESS) algorithm is shown in Algorithm \ref{alg:gess_alg}.
\begin{algorithm}
\SetKwData{Left}{left}\SetKwData{This}{this}\SetKwData{Up}{up}
\SetKwFunction{Union}{Union}\SetKwFunction{FindCompress}{FindCompress}
\SetKwInOut{Input}{Input}\SetKwInOut{Output}{Output}
\Input{$e_{j}(k),p_{j}(k),\; \forall k \in \{1,2,\cdots,N\};$}
\Output{$S_{j},\; B^{max}_{j};$}
\nl Calculate $S^{min}_{j}$ and $S^{max}_{j}$\;
\nl Assign $S_{j}=S^{min}_{j}$, derive $B^{max}_{j}$ and $f_{j}(S_{j},B^{max}_{j})$\;
\nl \While{$S^{max}_{j}\neq S^{min}_{j}$}{
\nl Assign $S^{tmp}_{j}=\lceil1/2(S^{max}_{j}+S^{min}_{j})\rceil$\;
\nl Calculate $B^{tmp}_{j}$ and $f_{j}(S^{tmp}_{j},B^{tmp}_{j})$ \;
\nl \If{$f_{j}(S^{tmp}_{j},B^{tmp}_{j})\leq f_{j}(S_{j},B^{max}_{j})$ or $f_{j}(S^{tmp}_{j},B^{tmp}_{j})=\inf$}{
\nl Assign $S_{j}=S^{tmp}_{j}$, $B^{max}_{j}=B^{tmp}_{j}$\;
\nl Assign $S^{min}_{j}=S^{tmp}_{j}$\;
}
\nl \Else{
\nl Assign $S^{max}_{j}=S^{tmp}_{j}$\;
}
}
\nl Assign $S_{j}=S^{min}_{j}$, derive $B^{max}_{j}$ and $f_{j}(S_{j},B^{max}_{j})$\;
\caption{The BESS algorithm \label{alg:gess_alg}}
\end{algorithm}
The binary search based algorithm requires at most $log_{2}(S^{max}_{j}-S^{min}_{j})$ iterations to find the optimal solar panel size. Within each iteration, calculating the battery capacity with a given solar panel size requires $N$ iterations. Thus, the complexity of the BESS algorithm is $O(N\log_{2}(S^{max}_{j}-S^{min}_{j}))$.

\section{Simulation Results}
\label{sec:simulation}

\begin{figure*}[!ht]
\centering
\hspace*{\fill}
  \begin{subfigure}[b]{0.5\textwidth}
   	\includegraphics[scale=0.3]{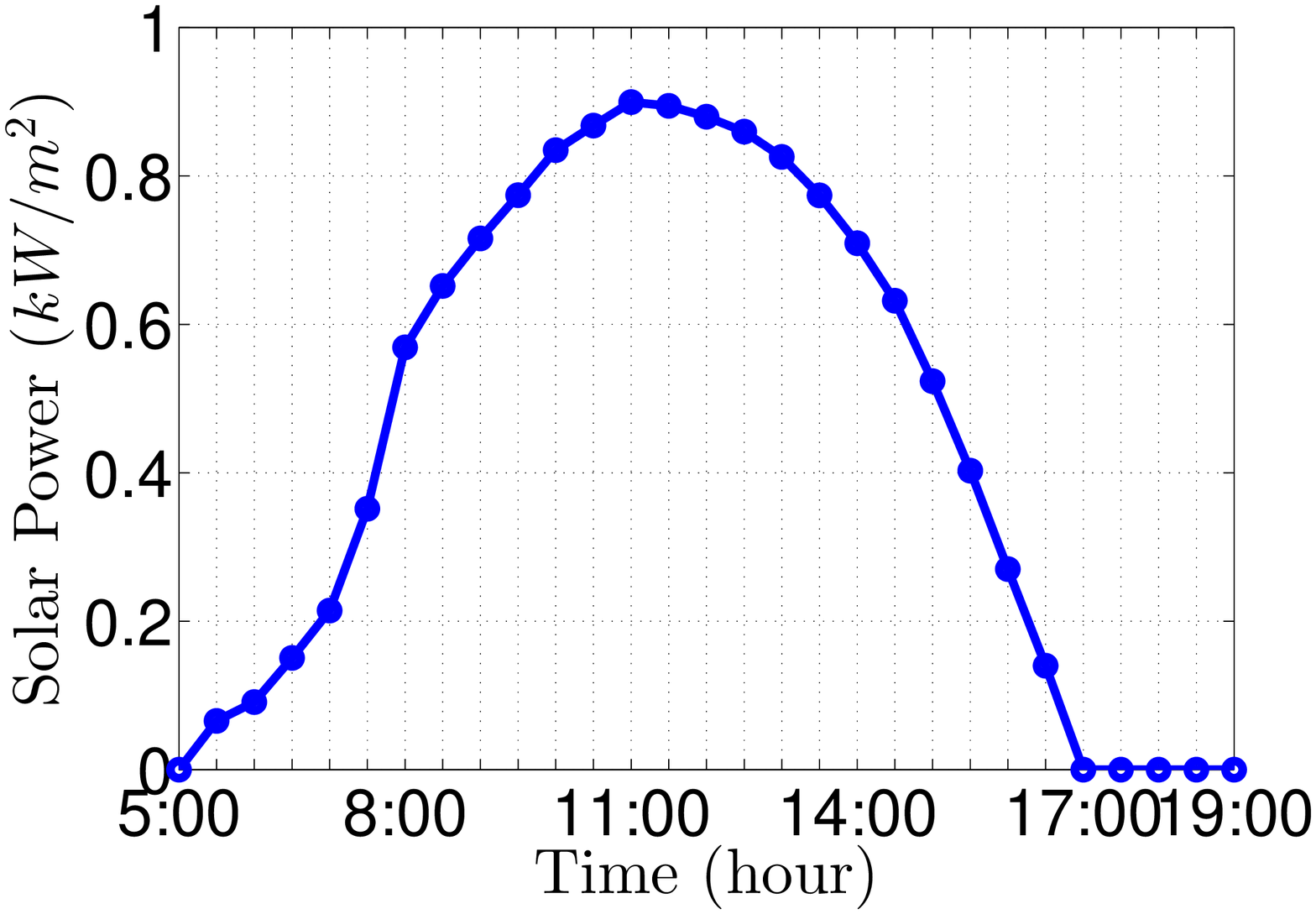}
	\caption{Solar power rate.}
	\label{fig:solar_power}
  \end{subfigure}%
  \begin{subfigure}[b]{0.5\textwidth}
   	\includegraphics[scale=0.3]{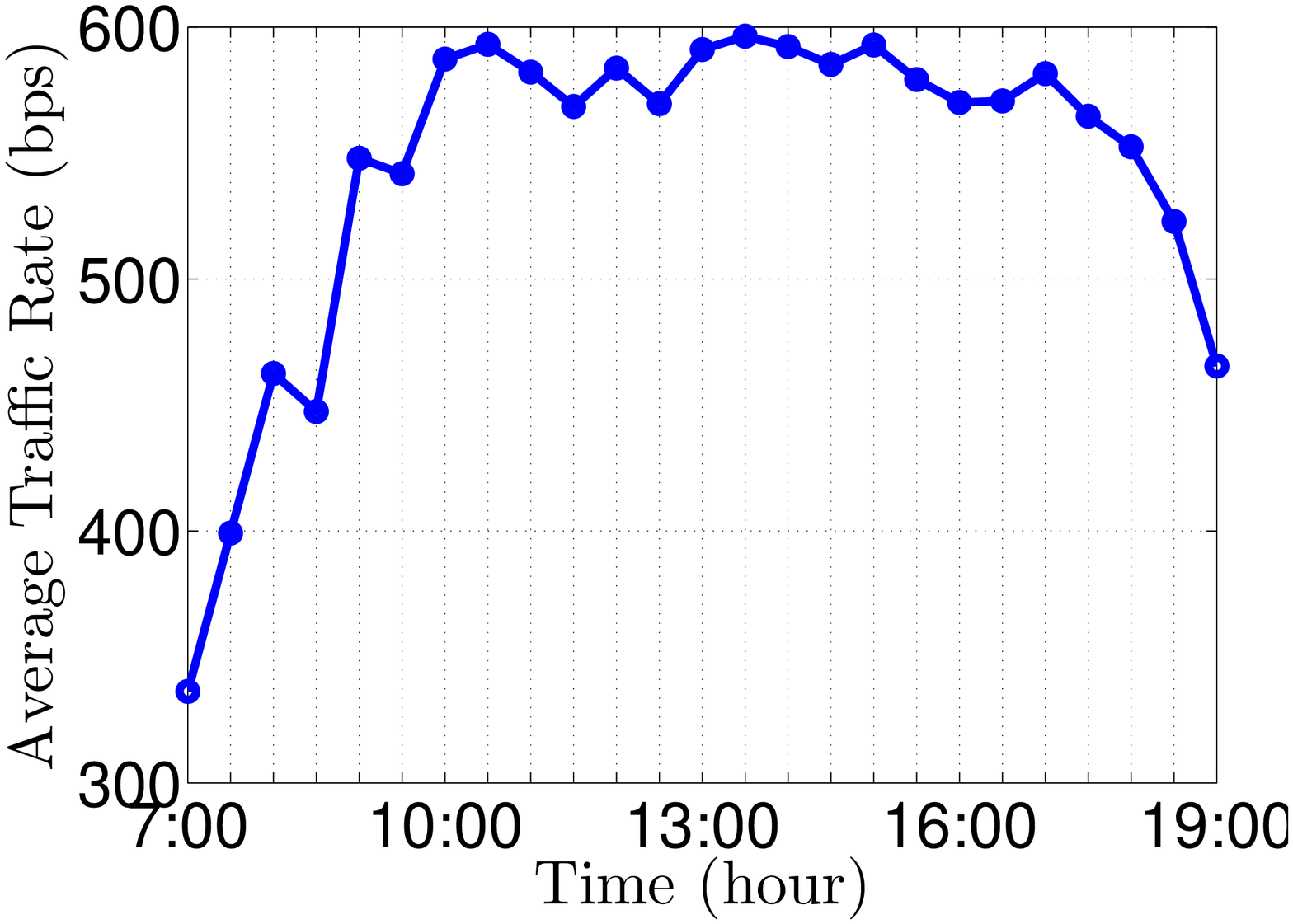}
	\caption{Average traffic rate.}
	\label{fig:average_traffic}
  \end{subfigure}
\end{figure*}

\begin{table}[ht]
\caption{Channel Model and Parameters}
\centering
\begin{tabular}{l||l}
\hline
Parameters & Value\\
\hline
$PL_{MBS}$ (dB) & $PL_{MBS}=128.1+37.6\log_{10}(d)$\\
$PL_{SCBS}$ (dB) & $PL_{SCBS}=38+10\log_{10}(d)$ \\
Rayleigh fading & 9 $dB$\\
Shadowing fading & 5 $dB$ \\
Antenna gain & 15 $dB$\\
Noise power level & -174 $dBm$ \\
Receiver sensitivity & -123 $dBm$ \\
\hline
\end{tabular}
\label{table:sim_parameters}
\end{table}

Simulations are set up to evaluate the performance of the proposed heuristic green energy provisioning solution in HetNets. In the simulation, we consider a HetNet with five MBSs and fifteen SBSs deployed in a $2km\times2km$ area. The MBS's transmission power is 43 $dBm$ while the SBS's transmission power is 33 $dBm$. The channel propagation model is based on COST 231 Walfisch-Ikegami \cite{COST231}. The model and parameters are summarized in Table \ref{table:sim_parameters}. Here, $PL_{MBS}$ and $PL_{SCBS}$ are the path loss between the users and MBSs and SCBSs, respectively. $d$ is the distance between users and BSs. The total bandwidth is 10~$MHz$ and the frequency reuse factor is one.

The static power consumption and the load-power coefficient of the MBS are 750~$W$ and 500, respectively \cite{Auer:2011:HME}. Here, we assume all MBSs have the same static power consumption and the same linear coefficient, $\beta_{j}=\beta, \forall j \in \mathcal{B}^{m}$. The duration of a time slot for the energy provisioning is 30 minutes. Solar power is utilized from 7:00 AM to 7:00 PM.
The solar power generation rate, which is shown in Fig.\ref{fig:solar_power}, is obtained from the UCSD solar resource web application \cite{UCSD:solar}. We generate the mobile traffic rates based on the mobile traffic pattern \cite{Peng:2011:TPS}. The average traffic rate at a location in the area is shown in Fig. \ref{fig:average_traffic}. We assume the green power percentage $\alpha$ is the same in all time slots. The MBSs' provisioning cost weights are randomly selected.

\begin{figure}[!ht]
\centering
	\includegraphics[scale=0.3]{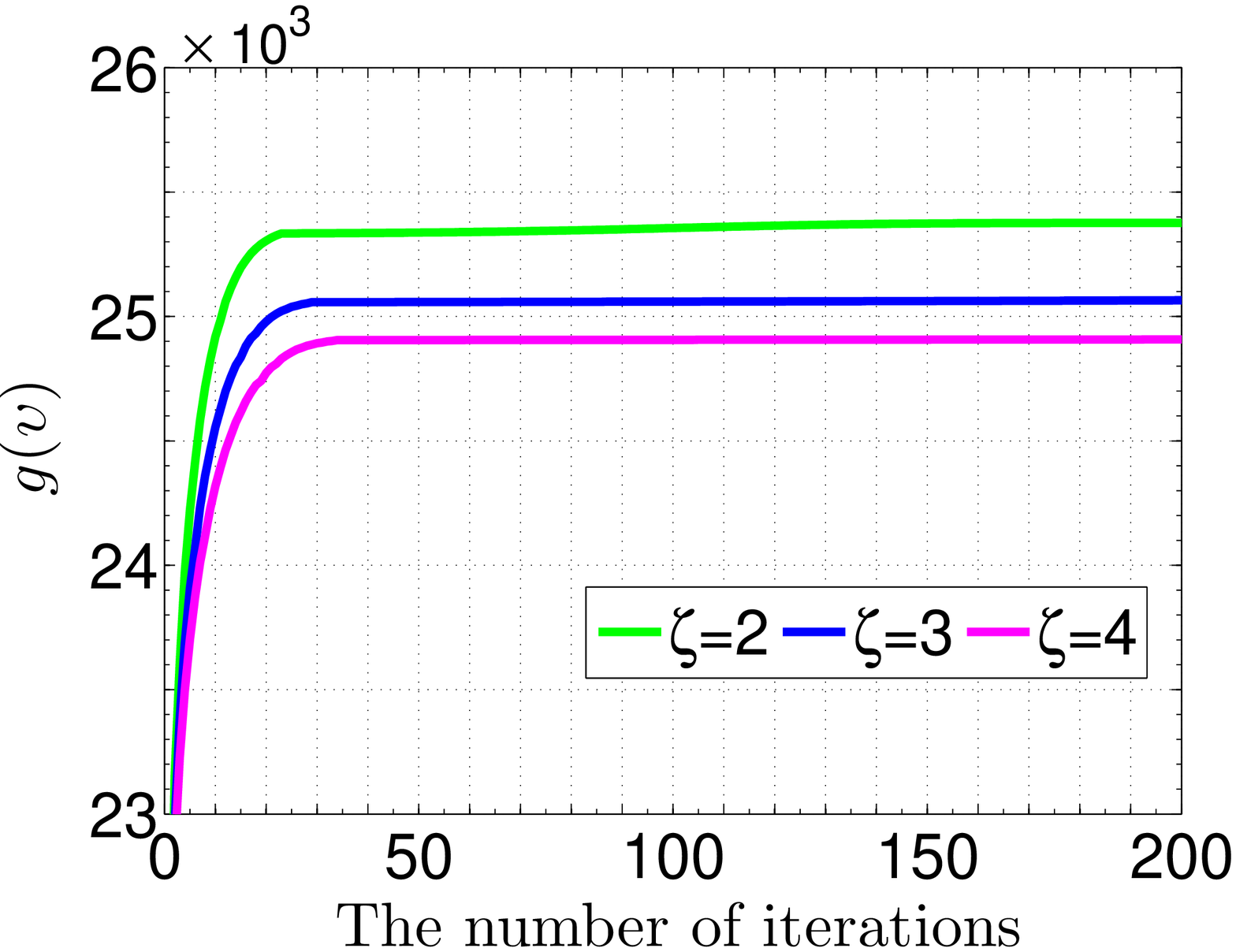}
	\caption{The converges of the provision cost aware load balancing ($\alpha=1$).}
	\label{fig:sim_1_converge}
\end{figure}

Fig. \ref{fig:sim_1_converge} shows the convergence of the provision cost aware (PCA) load balancing algorithm. The x-axis is the number of iterations between the traffic redirect algorithm and the traffic load update algorithm while the x-axis is the value of the dual function. After about fifty iterations, the value of the dual function converges. When $\zeta$ increases, the dual function converges to a smaller value. When the dual function converges to a smaller value, the primal function also has a smaller value. It indicates that increasing $\zeta$ reduces the provisioning costs. This is because when $\zeta$ increases, the network can tolerate additional traffic latency. As a result, more traffic load will be redirected to SBSs, thus reducing the power consumption of MBSs.

\begin{figure*}[!ht]
\centering
  \begin{subfigure}[t]{0.5\textwidth}
 	\includegraphics[scale=0.3]{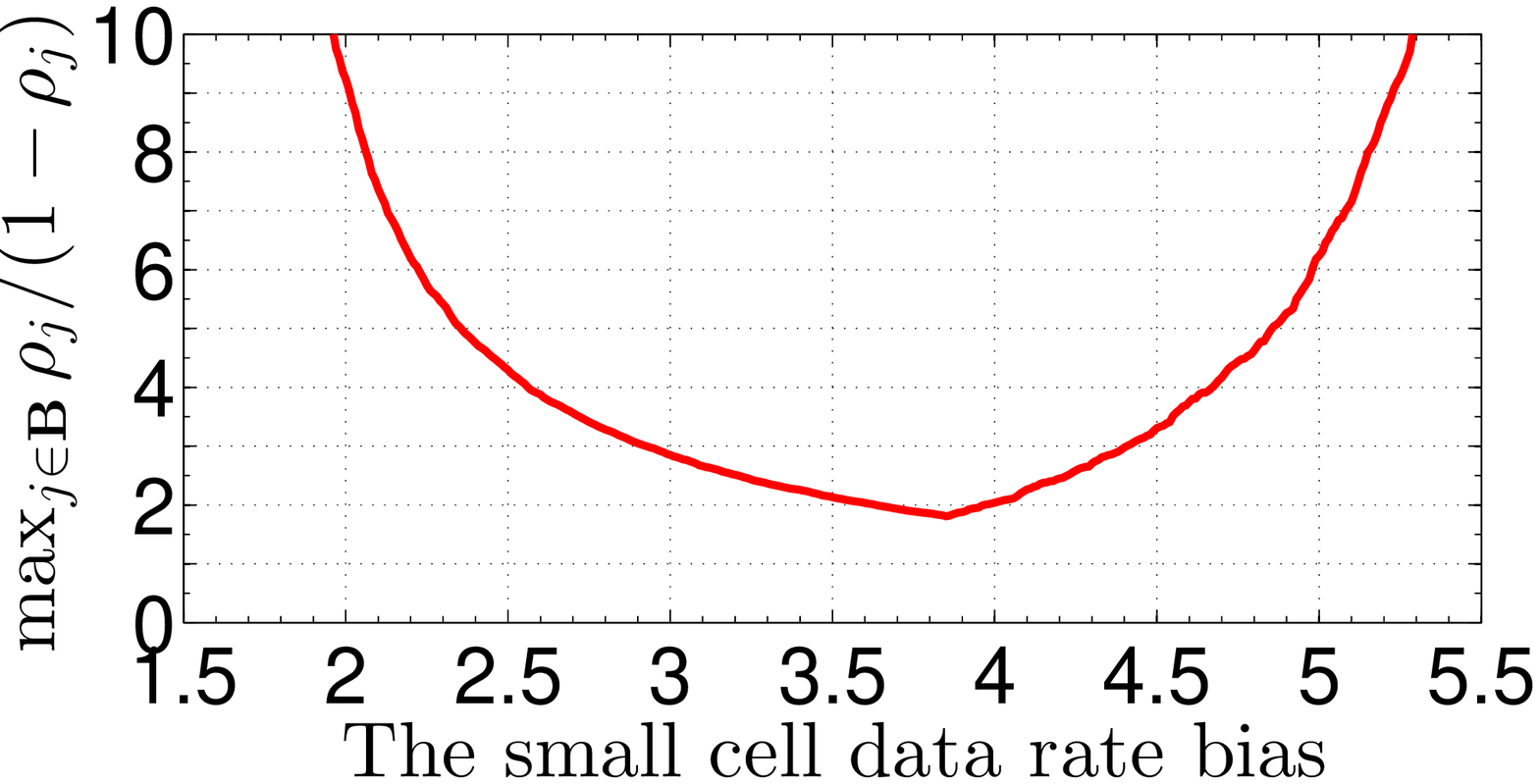}
	\caption{The maximum latency ratio.}
	\label{fig:sim_2_bias_latency}
  \end{subfigure}%
  \begin{subfigure}[t]{0.5\textwidth}
	\includegraphics[scale=0.3]{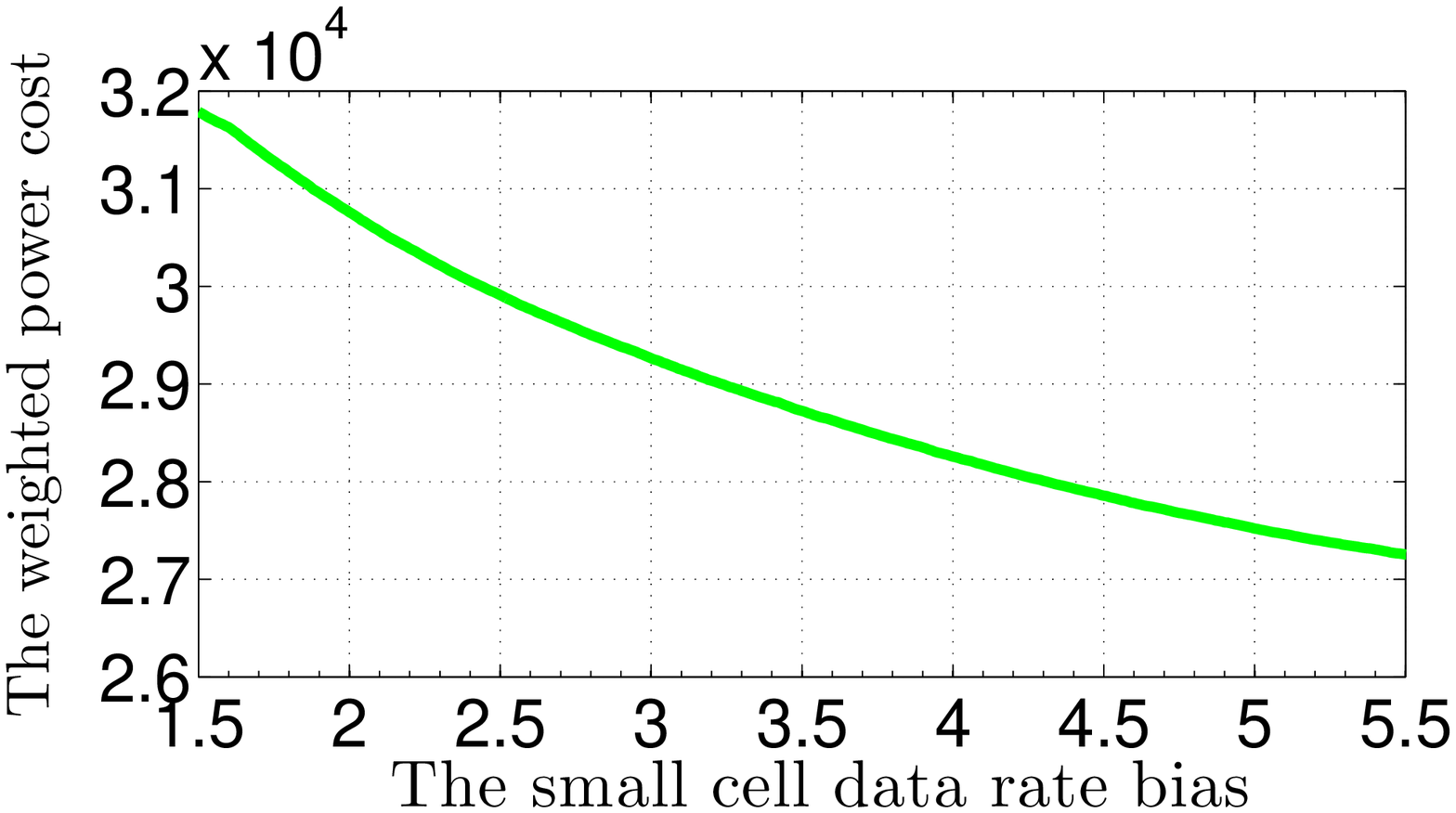}
	\caption{The weighted power cost.}
	\label{fig:sim_2_bias_power}
  \end{subfigure}
\end{figure*}

The traffic load balancing scheme is critical in minimizing the green energy provision cost. We compare the proposed PCA traffic load balancing scheme with the data rate bias (DRB) scheme \cite{Andrews:2014:AOLB} and the traffic latency minimization (LM) scheme.

In the simulation, we consider a two-tier data rate bias scheme and assume that BSs in the same tier have the same cell bias. Since different data rate bias leads to different traffic load balancing results, we first evaluate the two-tier data rate bias scheme and find a proper data rate bias. In the simulation, MBSs are in the first tier while SBSs are in the second tier. The cell bias of a MBS is one. We vary the cell bias of a SBS to investigate the performance of the scheme. In the data rate bias algorithm, a user selects the BS to maximize the biased data rate.
\begin{equation}
\label{eq:cre_alg}
b(x)=\arg\max_{j\in\mathcal{B}^{m}\cup\mathcal{B}^{s}} Z_{j}r_{j}(x).
\end{equation}
Here, $b(x)$ and $Z_{j}$ are the index of the selected BS and the cell bias of the $j$th BS, respectively.

Fig. \ref{fig:sim_2_bias_latency} shows that the maximum traffic latency ratio is a convex function of the data rate bias. The minimum value is achieved when the data rate bias is about 3.9. Meanwhile, Fig. \ref{fig:sim_2_bias_power} shows that the weighted power cost reduces as the data rate bias increases. This is because increasing the data rate bias allows more traffic offloaded to SBSs and thus reduces the power consumption of MBSs. In the simulation, since $\zeta$ equals to two, we set the data rate bias to four for comparing the PCA scheme. Note that when the data rate bias equals to four, the maximum traffic latency ration is around two.

The traffic latency minimization scheme solves the latency aware problem (LAP) as
\begin{eqnarray}
\label{eq:obj_lap}
\min_{\boldsymbol{\rho}} && \sum_{j \in \mathcal{B}^{m}\cup\mathcal{B}^{s}}L(\rho_{j})\\
\label{eq:constraint_lap}
subject\; to: && 0\leq\rho_{j}\leq 1-\epsilon.
\end{eqnarray}

\begin{figure*}
\centering
  \begin{subfigure}[t]{0.5\textwidth}
	\includegraphics[scale=0.3]{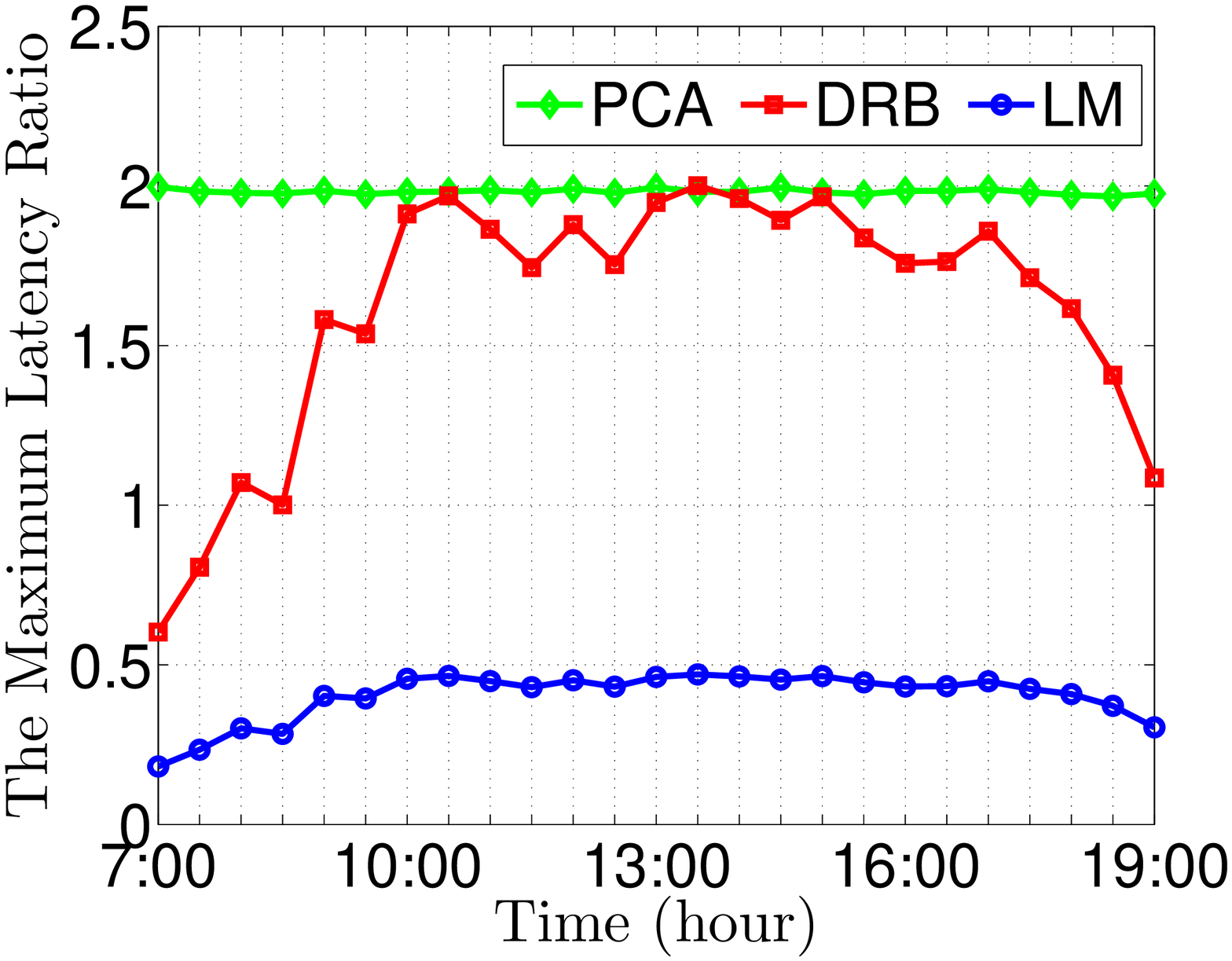}
	\caption{The maximum latency ratio of BSs ($\alpha=1$).}
	\label{fig:sim_3_latency}
  \end{subfigure}%
  \begin{subfigure}[t]{0.5\textwidth}
	\includegraphics[scale=0.3]{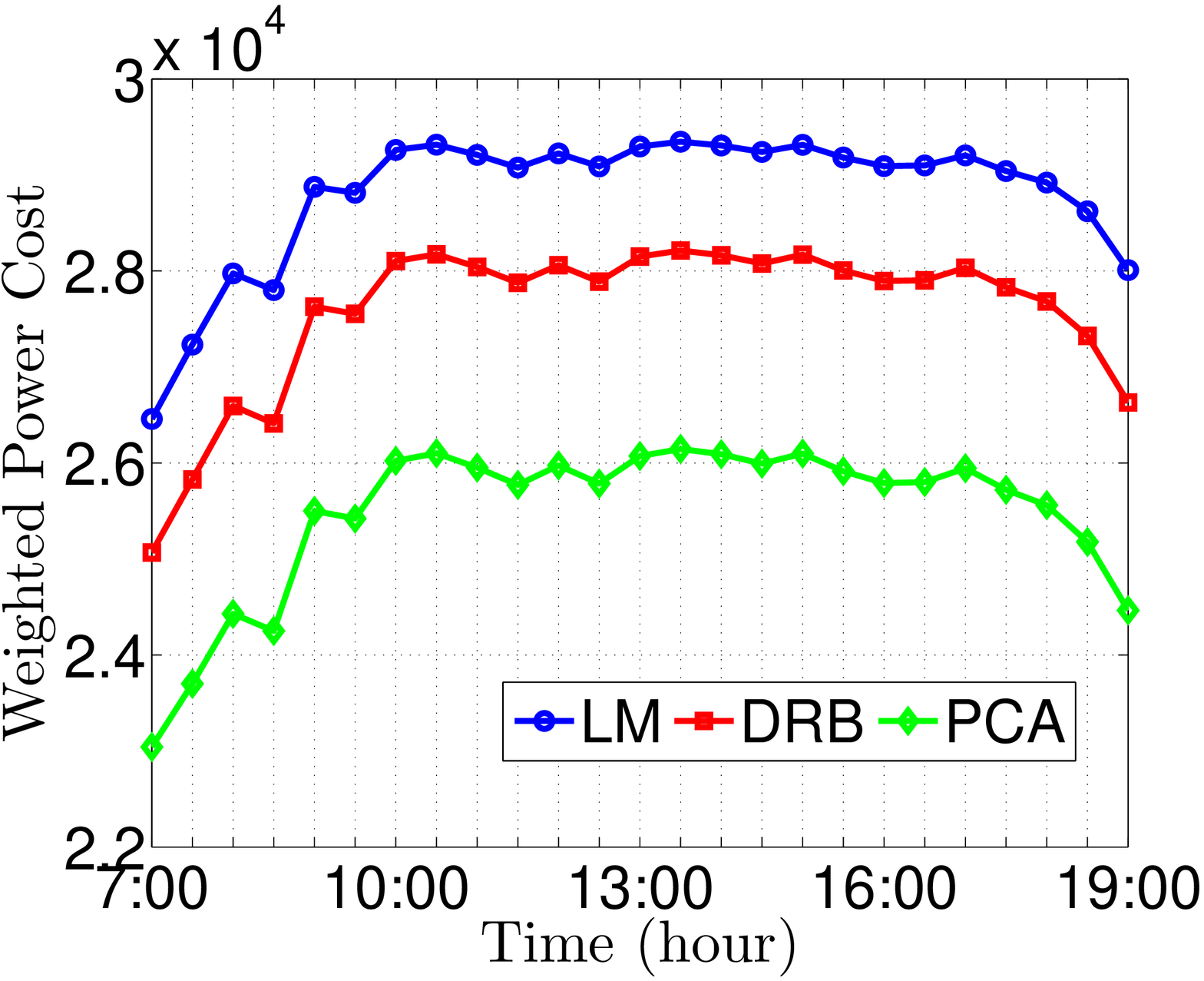}
	\caption{The weighted power cost of the network ($\alpha=1$).}
	\label{fig:sim_3_weighted_power}
  \end{subfigure}
\end{figure*}

Fig. \ref{fig:sim_3_latency} compares the maximum traffic latency ratio of the network under three traffic load balancing schemes. Since $\zeta$ equals to two, the PCA scheme maximizes the traffic offloading while ensuring $\zeta\leq 2$. The maximum traffic latency ratio of the data rate bias scheme depends on traffic intensity of the networks. When traffic intensity is low (high), the DRB scheme achieves a small (large) traffic latency ratio. This is because the data rate bias is fixed and the traffic balancing rule does not change over time slots. The LM schemes has minimal maximum traffic latency ratio.

Fig. \ref{fig:sim_3_weighted_power} shows the weighted power cost of the network under these traffic load balancing schemes. The PCA scheme has the minimal weighted power cost as compared to the other schemes. This is because the PCA scheme offloads as much traffic load as allowed by the QoS constraint to SBSs. In this way, the total power consumption of MBSs is reduced. In addition, the PCA scheme also balances the traffic load among MBSs according to their provision weights. The MBS with a large provision weight serves less traffic loads than the MBS with a small provision weight does. Although the PCA scheme has the highest traffic latency ratio, the QoS of the network is guaranteed. Moreover, the traffic latency ratio of the algorithm can be adjusted by adapting $\zeta$.

\begin{figure}[!ht]
\centering
	\includegraphics[scale=0.3]{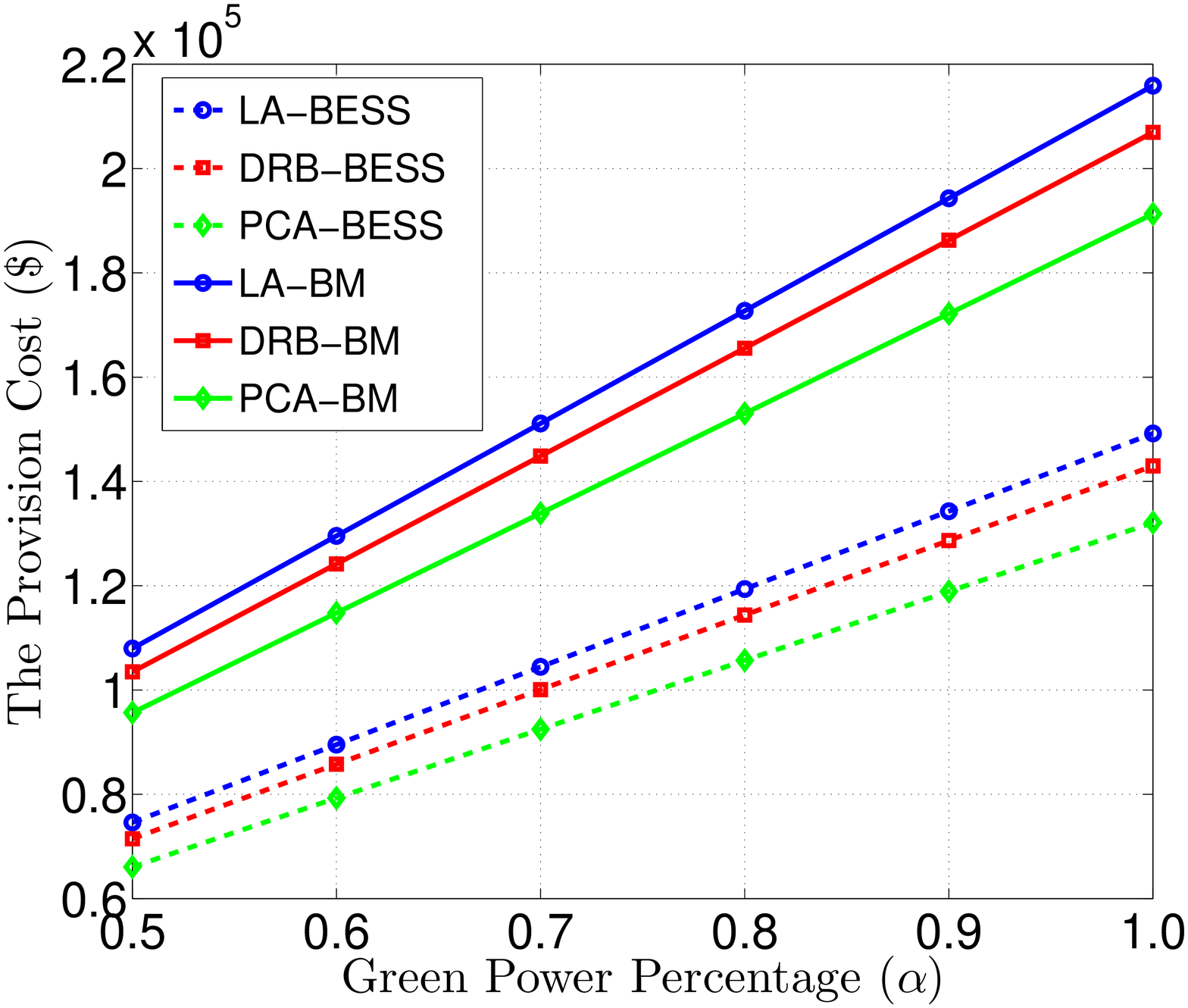}
	\caption{The total provision cost.}
	\label{fig:sim_4_total_cost}
\end{figure}

In Fig. \ref{fig:sim_4_total_cost}, we compare the total green energy provision costs of different solutions. The green energy provision solutions consist of two parts: the traffic load balancing scheme and the green energy system sizing scheme. In the simulation,  the per $m^{2}$ cost of the solar panel and the per watt costs of the battery are \$0.9 and \$0.2, respectively. For the traffic load balancing scheme, we adopt the PCA traffic load balancing scheme, the DRB traffic load balancing scheme, and the LM traffic load balancing scheme. For the green energy system sizing scheme, we compare the proposed BESS algorithm and a battery minimization (BM) sizing algorithm that minimizes the battery capacity. In the simulation, the proposed solution that consists of the PCA load balancing scheme and the BESS algorithm incurs the smallest provision cost. The provision cost of the network increases versus the green energy percentage. This is because a larger green energy percentage indicates more power should be pulled from the green energy generator, thus requiring a more powerful green energy system.

\section{Conclusion}
\label{sec:conclusion}
In this paper, we have proposed a green energy provisioning solution to minimize the CAPEX of deploying the green energy system for MBSs in HetNets while achieving the targeted QoS requirement. The green energy provisioning solution consists of the provision cost aware traffic load balancing algorithm and the binary energy system sizing algorithm. Given the traffic load, the provision cost aware traffic load balancing algorithm balances the traffic load among BSs based on the QoS requirements and the provision costs. The energy consumption of MBSs are calculated based on their traffic loads. The BESS algorithm optimizes the solar panel sizes and battery capacities for individual MBSs based on their power consumption. The simulation results have validated the performance and the viability of the proposed solution.

Although various traffic load balancing algorithms may be adopted in HetNets, the energy provision solution based on provision cost aware traffic load balancing provides a lower bound on the provision costs of the green energy systems. The results provide guidance for the network planning and deployments from the perspective of provisioning green energy in cellular networks..

\bibliographystyle{IEEEtran}
\bibliography{mybib}
\end{document}